\documentclass[conference,10pt]{IEEEtran}

\usepackage{times}
\usepackage{amsmath}
\usepackage{amssymb}
\usepackage{algorithm}
\usepackage{algorithmic}
\usepackage{color, array, colortbl}
\usepackage{graphicx}
\usepackage{multirow}
\usepackage{comment}
\usepackage{framed,color}
\definecolor{shadecolor}{rgb}{0.9,0.9,0.9}

\newcommand{\G}{\mathcal{G}}

\newcommand{\COMMENTED}[1]{}
\newcommand{\ALG}{\textsc{Alg}}
\newcommand{\OPT}{\textsc{Opt}}
\newcommand{\Cost}{\text{Cost}}

\newcommand{\size}{\text{size}}
\newcommand{\Epoch}{\mathcal{E}}
\newcommand{\BW}{\omega}
\newcommand{\LAT}{\lambda}

\newcommand{\counter}{\textsc{C}}
\newcommand{\counterLarge}{\textsc{C}$_L$}
\newcommand{\RAND}{\textsc{Mix}}
\newcommand{\DET}{\textsc{Cen}}
\newcommand{\PRAND}{\textsc{Mix}}
\newcommand{\PDET}{\textsc{Cen}}
\newcommand{\STAT}{\textsc{Stat}}

\newtheorem{theorem}{Theorem}[section]

\newtheorem{lemma}[theorem]{Lemma}

\usepackage{algorithm}
\usepackage{algorithmic}

\newcommand{\ignore}[1]{}

\begin{document}





\title{Online Strategies for Intra and Inter Provider\\Service Migration in Virtual Networks}

\author {
   Dushyant Arora$^1$, Marcin Bienkowski$^2$, Anja Feldmann$^1$,
   Gregor Schaffrath$^1$, Stefan
   Schmid$^1$\\
\small{$^1$ Deutsche Telekom Laboratories / TU Berlin, Germany}\\
\small{\texttt{\{darora,anja,grsch,stefan\}@net.t-labs.tu-berlin.de}}\\
\small{$^2$ Institute of Computer Science, University of Wroc{\l}aw,
Poland};
\small{\texttt{mbi@ii.uni.wroc.pl}}\\
}

\date{}

\maketitle

\sloppy


\begin{abstract}
Network virtualization allows one to build dynamic distributed
systems in which resources can be dynamically allocated at locations
where they are most useful. In order to fully exploit the benefits
of this new technology, protocols need to be devised which react
efficiently to changes in the demand. This paper argues that the
field of online algorithms and competitive analysis provides useful
tools to deal with and reason about the uncertainty in the request
dynamics, and to design algorithms with provable performance
guarantees.

As a case study, we describe a system (e.g., a gaming application)
where network virtualization is used to support thin client
applications for mobile devices to improve their QoS. By decoupling
the service from the underlying resource infrastructure, it can be
migrated closer to the current client locations while taking into
account migration cost. This paper identifies the major cost factors
in such a system, and formalizes the corresponding optimization
problem. Both randomized and deterministic, gravity center based
online algorithms are presented which achieve a good tradeoff
between improved QoS and migration cost in the worst-case, both for
service migration within an infrastructure provider as well as for
networks supporting cross-provider migration. The paper reports on
our simulation results and also presents an explicit construction of
an optimal offline algorithm which allows, e.g., to evaluate the
competitive ratio empirically.
\end{abstract}

%
%

\section{Introduction}
\label{sec_intro}

In 2008, the total number of mobile web users outgrew the total
number of desktop computers with respect to Internet
users~\cite{itu09} for the first time.  Providing high
quality-of-service (QoS) respective an excellent quality of
experience (QoE) to mobile Internet clients is very challenging,
e.g., due to user mobility. However, many applications, including
such popular applications as gaming, need a reliable, continuous
network service with minimal delay.

\emph{Network virtualization}~\cite{virsurvey} is an emerging
technology which allows a service specific network to be embedded
onto a substrate network in a dynamic fashion.  This includes
migration of virtual nodes and links as well as virtual servers to
meet the applications demands for connectivity and performance. To
take full advantage of this flexibility it is often necessary to
know future application demands. Yet, this information is typically
not readily available and therefore neither the network resources
can be used in an optimal manner nor do the users receive the best
possible service.

This paper studies a mobile thin client application (such as a
\emph{game server}~\cite{oss}) that is supported by network
virtualization technology. We assume that the distribution of thin
clients and therefore the request pattern changes over time. For
instance, at 1 a.m.~GMT many requests may originate in Asian
countries, then more and more requests come from European users and
later from the United States. In this setting it can be beneficial
to migrate (or \emph{re-embed}) the service closer to the users,
e.g., to minimize access delays for the users and to minimize
network costs for the providers~\cite{visa09mig}. Network
virtualization allows us to realize such networks.

While moving services close to clients can reduce latency, migration
also comes at a cost: the bulk-data transfer imposes load on the
network and may cause a service disruption. In particular, the cost
of migration depends on the available bandwidth in the substrate
network~\cite{vee07}. Moreover, if virtual networks (VNets) are
provisioned across administrative domains belonging to multiple
infrastructure providers, (inter-provider) migration entails certain
transit (or \emph{roaming}) costs.

To gain insights into this tradeoff, we identify the main costs
involved in this system. Intuitively, the benefits from
virtualization are higher the lower the migration cost is relative
to the latency penalty. Therefore, a predictable access pattern may
ease migration. However, in practice, user arrival patterns are hard
to predict, and thus we, in this paper, explicitly incorporate
uncertainty about future arrivals.

The classic formal tool to study algorithms that deal with inputs
(or more specifically: request accesses) that arrive in an online
fashion and cannot be predicted is the \emph{competitive analysis
framework}. In competitive analysis, the performance of a so-called
online algorithm is compared to an optimal offline algorithm that
has complete knowledge of the input \emph{in advance}. In effect,
the competitive analysis is a \emph{worst case performance analysis}
that does not rely on any statistical assumptions. We apply this
framework to network virtualization and propose---for a simplified
model where, e.g., the main access cost is delay to the server and
the main migration cost is the available bandwidth between migration
source and destination---a competitive migration algorithm whose
performance is close to the one of the optimal offline algorithm.

\subsection{Related Work} \label{sec_relwork}

\textbf{Mobile Networks.} The mobile web today provides
browser-based access to the Internet or web applications to millions
of users connected to a wireless network via a mobile device. There
exists a vast amount of work on the subject, and we refer the reader
to the introductory books, e.g.,~\cite{mbook2}. In this project, we
tackle the question of how \emph{network virtualization} can be used
to improve the quality of service for mobile devices.

\textbf{Network Virtualization.} Network virtualization has gained a
lot of attention recently~\cite{geni} as it enables the co-existence
of innovation and reliability~\cite{visa09virtu} and promises to
overcome the ``ossification'' of the Internet~\cite{ossification}.
For a more detailed survey on the subject, please refer
to~\cite{virsurvey}. Virtualization allows to support a variety of
network architectures and services over a shared substrate, that is,
a \emph{Substrate Network Provider (SNP)} provides a common
substrate supporting a number of \emph{Diversified Virtual Networks
(DVN)}. OpenFlow~\cite{openflow} and VINI~\cite{Bavier06} are two
examples that allow researchers to (simultaneously) evaluate
protocols in a controllable and realistic environment.
Trellis~\cite{trellis08} provides a software platform for hosting
multiple virtual networks on shared commodity hardware and can be
used for VINI. Network virtualization is also useful in data center
architectures, see, e.g., \cite{secondnet}.

\textbf{Embedding.} A major challenge in network virtualization is
the \emph{embedding}~\cite{embedding} of VNets, that is, the
question of how to efficiently and on-demand assign incoming service
requests onto the topology. Due to its relevance, the embedding
problem has been intensively studied in various settings, e.g., for
an offline version of the embedding problem see~\cite{turner}, for
an online and competitive algorithm see~\cite{moti10tr}, for an
embedding with only bandwidth constraints see~\cite{ammar}, for
heuristic approaches without admission control see~\cite{zhu06}, or
for a simulated annealing approach see~\cite{simannealing}. Since
the general embedding problem is computationally hard, Yu et
al.~\cite{rethinking} advocate to rethink the design of the
substrate network to simplify the embedding; for instance, they
allow to split a virtual link over multiple paths and perform
periodic path migrations. Lischka and Karl~\cite{pb-embed} present
an embedding heuristic that uses backtracking and aims at embedding
nodes and links concurrently for improved resource utilization. Such
a concurrent mapping approach is also proposed in~\cite{infocom2009}
with the help of a mixed integer program (a more general formulation
which also includes migration aspects can be found
in~\cite{emb10TR}). Finally, several challenges of embeddings in
wireless networks have been identified by Park and
Kim~\cite{wembed}.

In contrast to the approaches discussed above we, in this paper,
tackle the question of how to dynamically embed or migrate virtual
servers~\cite{mobitopolo} in order to efficiently satisfy connection
requests arriving online at any of the network entry points, and
thus use virtualization technology to improve the quality of service
for mobile nodes. The relevance of this subproblem of the general
embedding problem is underlined by Hao et al.~\cite{visa09mig} who
show that under certain circumstances, migration of a Samba
front-end server closer to the clients can be beneficial even for
bulk-data applications.

\textbf{Online Algorithms.} To the best of our knowledge,
\cite{visa10} and \cite{hotice11} (for online server migration), and
\cite{moti10tr} (for online virtual network embeddings) are the only
works to study network virtualization from an online algorithm
perspective. The formal competitive migration problem is related to
several classic optimization problems such as facility location,
$k$-server problems, or online page migration. All these problems
are a special case of the general \emph{metrical task system}
(e.g.,~\cite{borodin,metricaltask}) for which there is, e.g., an
asymptotically optimal deterministic $\Theta(n)$-competitive
algorithm, where $n$ is the state (or ``configuration'') space; or a
randomized $O(\log^2 n\cdot \log\log{n})$-competitive algorithm
given that the state space fulfills the triangle inequality: this
algorithm uses a (well separated) tree approximation for the general
metric space (in a preprocessing step) and subsequently solves the
problem on this distorted space; unfortunately, both algorithmic
parts are rather complex.

In the field of \emph{facility location}, researchers aim at
computing optimal facility locations that minimize building costs
and access costs (see, e.g.,~\cite{competitivefl} for an online
algorithm). In \cite{georgiosmigration}, Laoutaris et al.~propose a
heuristic algorithm for a variant of a facility location problem
which allows for facility migration; this algorithm uses
neighborhood-limited topology and demand information to compute
optimal facility locations in a distributed manner. In contrast to
our work, the setting is different and migration cost is measured in
terms of hop count. There is no performance guarantee. In the field
of \emph{$k$-server problems} (e.g.,~\cite{borodin}), an online
algorithm must control the movement of a set of $k$ servers,
represented as points in a metric space, and handle requests that
are also in the form of points in the space. As each request
arrives, the algorithm must determine which server to move to the
requested point. The goal of the algorithm is to reduce the total
distance that all servers traverse. In contrast, in our model it is
possible to access the server remotely, that is, there is no need
for the server to move to the request's position. The \emph{page
migration problem} (e.g., \cite{dynpmsurvey}) occurs in managing a
globally addressed shared memory in a multiprocessor system. Each
physical page of memory is located at a given processor, and memory
references to that page by other processors are charged a cost equal
to the network distance. At times, the page may migrate between
processors, at a cost equal to the distance times a page size
factor. The problem is to schedule movements on-line so as to
minimize the total cost of memory references. In contrast to these
page migration models, we differentiate between access costs that
are determined by latency and migration costs that are determined by
network bandwidth.

There is an intriguing relationship between server migration and
\emph{online function tracking}~\cite{swat10,soda09}. In online
function tracking, an entity Alice needs to keep an entity Bob
(approximately) informed about a dynamically changing function,
without sending too many updates. The online function tracking
problem can be transformed into a chain network where the function
values are represented by the nodes on the chain, and a sequence of
value changes corresponds to a request pattern on the chain. In
particular, it follows from~\cite{swat10} that already for some very
simple linear substrate networks of size $n = \Theta(\beta)$, where
$\beta$ is the migration cost, no deterministic or randomized online
algorithm can achieve a competitive ratio smaller than $\Omega(\log
n / \log \log n)$.

\textbf{Other.} A preliminary version of this paper was presented at
the VISA workshop~\cite{visa10} (a discussion of multiple server
scenarios appeared at the Hot-ICE workshop~\cite{hotice11}). We
extend the results in~\cite{visa10} in several respects: First, we
present a deterministic alternative to the randomized migration
strategy derived in~\cite{visa10}; this deterministic online
algorithm is based on gravity centers and achieves a better
competitive ratio against adaptive adversaries. Second, we extend
our model to settings with multiple infrastructure providers, and
describe competitive intra and inter provider migration strategies.
Provisioning virtual network services across multiple providers is
an interesting topic which is hardly explored in literature so far;
notable exceptions are \emph{PolyVINE}~\cite{polyvine}, a
distributed coordination protocol to perform cross-provider
embeddings, and \emph{V-Mart}~\cite{vmart} that describes an auction
framework for task partitioning. Third, we conducted extensive
experiments to complement our formal analysis.

\subsection{Contributions and Organization}

This paper studies a mobile network virtualization architecture
where thin clients on mobile devices access a service that can be
migrated closer to the access points to reduce user latency. We
identify the main costs in this system (Section~\ref{model}) and
introduce an optimization problem accordingly. Both for networks
with a single provider (Section~\ref{sec:intra}) as well as for
networks with multiple providers (Section~\ref{sec:inter}), online
migration strategies are presented that are provably competitive to
an optimal offline algorithm, i.e., that achieve a good performance
even in the worst-case. This paper also describes an optimal offline
algorithm which is useful for finding optimal strategies at
hindsight, for dealing with regular and periodic request patterns,
and for computing the competitive ratios of online algorithms
(Section~\ref{sec:optoff}). We report on our experiments in
different scenarios in Section ~\ref{sec:simulations}. In
Section~\ref{sec:conclusion}, the paper concludes.

\section{Architecture}\label{model}

Our work is motivated by the virtualization architecture proposed
in~\cite{visa09virtu} for which we are in the process of developing
a prototype implementation. The main roles of this architecture
related to this work are: The \emph{(Physical) Infrastructure
Provider (PIP)}, which owns and manages an underlaying physical
infrastructure called ``substrate'' (we will treat the terms
\emph{infrastructure provider} and \emph{substrate provider} as
synonyms); the \emph{Virtual Network Provider (VNP)}, which provides
bit-pipes and end-to-end connectivity to end-users; and the
\emph{Service Provider (SP)}, which offers application, data and
content services to end-users.

We assume that a service provider is offering a service to mobile
clients which can benefit from the flexibility of network and
service virtualization. The goal of the service provider is to
minimize the round-trip-time of its service users to the servers, by
triggering migrations depending, e.g., on (latency) measurements.
Concretely, VNP and/or PIPs will react on the SP-side changes of the
requirements on the paths between server and access points, and
re-embed the servers accordingly.

In the remainder of this paper, if not stated otherwise, the term
\emph{provider} will refer to the PIP role in the above
architecture. In particular, we will study multi provider scenarios
where the service provider may decide to migrate an application
across PIP boundaries.

\subsection{General Cost Model}

Formally, we consider a substrate network $G=(V,E)$ managed by one
or multiple substrate providers (PIP). Each substrate node $v\in V$
has certain properties and features associated with it (e.g., in
terms of operating system or CPU power); in particular, we assume
that it has a computational capacity $c(v)$. Similarly, each link
$e=(u,v)\in E$, with $u,v\in V$, has certain properties, e.g., it is
characterized by a bandwidth capacity $\BW(e)$, and it offers the
latency $\LAT(e)$. Links between different PIPs are typically more
expensive than links within a PIP.

In addition to the substrate network, there is a set $T$ of external
machines (the mobile thin clients or simply \emph{terminals}) that
access $G$ by issuing requests to virtualized services hosted on a
set of virtual servers by $G$. There is a set of services
$\mathcal{S}=\{S_1, S_2, ...\}$ where each service $S_i$ can be
offered by multiple servers $s\in S_i$. (In the technical part of
the paper we will focus on a single-server scenario only.) Each
server $s$ has a certain resource or capacity requirement $r(s)$
that needs to be allocated to $s$ on the substrate node where it is
hosted.

In order for the machines in $T$ to access the services
$\mathcal{S}$, a fixed subset of nodes $A\subseteq V$ serve as
\emph{Access Points} where machines in $T$ can connect to $G$. Due
to the movement of machines in $T$, the access points can change
frequently, which may trigger the migration algorithm. We define
$\sigma_t$ to be the multi-set of requests at time $t$ where each
element is a tuple $(a\in A, S\in \mathcal{S})$ specifying the
access point and the requested service $S$. (For ease of notation,
when clear from the context, we will sometimes simply write $v\in
\sigma_t$ to denote the multi-set of access points used by the
different requests.) Our main objective is to shed light onto the
trade-off between the access costs $\Cost_{\text{acc}}$ of the
mobile clients to the current service locations and the server
migration cost $\Cost_{\text{mig}}$: while moving the servers closer
to the requester may reduce the access costs and hence improve the
quality of service, it also entails the overhead of migration.


We can identify the following main parameters which influence the
access and migration costs. A major share of $\Cost_{\text{acc}}$ is
due to the request latency, i.e., the sum of the requests' latencies
to the corresponding servers. Observe that the routing of the
requests occurs along the shortest paths (w.r.t.~latency) on the
substrate network. In addition, the access cost depends on the
server load, that is, the access cost depends on the capacity $c(v)$
of the hosting node $v$ and the resource demands $r(s)$ of the
servers $s$ hosted by $v$. The correlation between load and delay
can be captured by different functions, and is not studied further
here. In this paper, we assume that requests are relatively small,
and hence, we do not explicitly model bandwidth constraints in
$\Cost_{\text{acc}}$. In conclusion, at time $t$ and for some
function $f$,
$$
\Cost_{\text{acc}}(t) = \sum_{r_t\in \sigma_t}
f\left(\text{delay}(r_t), \text{load}(r_t)\right).
$$

In contrast to the requests, which are rather light-weight, the
server state is typically large, and hence the traffic volume of
migration cannot be neglected. The main cost of migration are
service outage periods and the migration itself. The migration cost
$\Cost_{\text{mig}}$ of a virtual server $s\in S$, or the outage
period, hence depends to a large extent on the available bandwidth
$\BW(p)$ on the migration path $p: src \rightsquigarrow dst$ (along
the substrate network) between migration source node $src$ and
destination node $dst$, and the size $\text{size}(s)$ of the
application $s$ to be migrated. Another major cost factor is the
\emph{transit costs}, namely the number $k$ of PIPs on the path.
In summary:
$$
\Cost_{\text{mig}}(t) = \sum_{s\in S} f(\BW(p),k,\text{size}(s))
$$
for some function $f$, where the migration cost is zero if
$src=dst$.

Our model so far lacks one additional ingredient: \emph{terminal
dynamics} (or mobility). One approach would be to assume arbitrary
request sets $\sigma_t$, where $\sigma_t$ is completely independent
of $\sigma_{t-1}$. However, for certain applications it may be more
realistic to assume that the mobile nodes move ``slowly'' between
the access points. Note that while users typically travel between
different cities or countries at a limited speed, these geographical
movements may not translate to the topology of the substrate
network. Thus, rather than modeling the users to travel along the
links of $G$, we consider on/off models where a user appears at some
access point $a_1\in A$ at time $t$, remains there for a certain
period $\Delta t$, before moving to another arbitrary node $a_2\in
A$ at time $t+\Delta t$.

One may assume that $\Delta t$ is exponentially distributed.
However, in our formal analysis we assume a worst-case perspective
and consider arbitrary distributions for $\Delta t$. Often, it is
reasonable to assume some form of correlation between the individual
terminals' movement. For example, in an urban area, workers commute
downtown in the morning and return to suburbs in the evening. Or in
a planetary-scale substrate network, demographic aspects have to be
taken into account in the sense that during a day, first many
requests will originate from Asian countries, followed by an active
period in Europe and finally America. However, as it is rather hard
to describe and characterize such movement accurately we, in the
formal part of this paper, perform a worst case analysis
(w.r.t.~latency) that does not use any statistical assumptions.

To what extent the system can benefit from virtual network support
and migration depends on several factors, e.g., how frequently the
thin clients change the access points. Given rapid changes it may be
best to place the server in the middle of the network and leave it
there. On the other hand, if the changes are slower or can be
predicted, it can be worthwhile to migrate the server to follow the
mobility pattern. This constitutes the trade-off studied in this
paper.

\subsection{Competitive Analysis}\label{companal}

As already discussed, competitive analysis asks the question: How
well does the system perform compared to an optimal offline strategy
which has complete knowledge of the entire request sequence in
advance? In the following, we present an online migration strategy
that is ``competitive'' to any other online or offline solution for
virtual network supported server migration. In order to focus on the
main properties and trade-offs involved in the virtualization
support of thin clients, we assume a simplified online framework for
our formal analysis. We consider a synchronized setting where time
proceeds in time slots (or \emph{rounds}).\footnote{Note that while
this assumption simplifies the analysis, it is not critical for our
results.} In each round $t$, a set of $\sigma_t$ terminal requests
arrive in a worst-case and online fashion at an arbitrary set of
access nodes $A$.

Thus the embedding problem is equivalent to the following
synchronous game, where an online algorithm $\ALG$ has to decide on
the migration strategy in each round $t$, without knowing about the
future access requests.  In each round $t\geq 0$:
\begin{enumerate}

\item[1.] The requests $\sigma_t$ arrive at some access nodes $A$.

\item[2.] The online algorithm $\ALG$ decides where in $G$ to migrate the
servers $S$. If positions are changed, it pays migration costs
$\Cost_{\text{mig}}(t)$.

\item[3.] The online algorithm $\ALG$ pays the requests' access costs $\Cost_{\text{acc}}(t)$ to the corresponding servers (e.g, hop distance).
\end{enumerate}

Note, that we allow $\ALG$ to migrate the virtual servers for all
the requests of the current time slot $t$. However, as we assume
that a request is much cheaper than a migration, and if there are
not too many requests arriving concurrently, our results also apply
to scenarios where the last two steps are reordered.

We aim at devising competitive algorithms $\ALG$ that minimize the
\emph{competitive ratio} $\rho$:
Let $\ALG_t(\sigma)$ be the migration and access costs incurred by
$\ALG$ in round $t$ under a request arrival sequence $\sigma$ (a
sequence of access points), that is,
$$
\ALG_t(\sigma)=\Cost_{\text{acc}}(t)+ \Cost_{\text{mig}}(t).
$$

Let $\OPT(\sigma)$ be the optimal cost of an offline algorithm
$\OPT$ for the given $\sigma$, that is, $\OPT$ has a complete
knowledge of $\sigma$ and can hence optimize the server locations
``offline''. $\rho$ is the ratio of the costs of $\ALG$ and $\OPT$.
Thus, our objective is to minimize:
$$
\rho = \max_{\sigma} \frac{\sum_t \ALG_t(\sigma)}{\OPT(\sigma)}
$$
In case of online algorithms that use randomization, we consider the
expected costs against an oblivious adversary without access to the
outcome of the random coin flips of the algorithm.

For our analysis, we make the following simplification: The
migration cost $\Cost_{\text{mig}}(t)$ is given by the bandwidth
constraint of the smallest edge capacity on the migration path, plus
the number of PIPs traversed times $\pi$. Let
$\Cost_{\text{mig}}(u,v)$ denote the migration cost on a path from
$u$ to $v$. (The path will be clear from the context.) Thus,
$\Cost_{\text{mig}}(u,v)=\max_e \size(s)/\BW(e)+k\cdot\pi$ where
$\size(s)$ is the size of the migrated server $s$, $e$ is a link on
the migration path from $u$ to $v$, $k$ is the number of PIPs
traversed and $\pi$ is the cost of migrating across a PIP boundary.
Observe that given a migration path, the cost is different in case
the migration occurs once along the entire path compared to the case
where it occurs in two steps at different times.

\section{Competitive Intra-Provider Migration}\label{sec:intra}

This section presents two online protocols for server migration
within a single provider (PIP). The main idea of both algorithms is
to divide time into \emph{epochs}: As we will see, also an optimal
offline algorithm will have certain costs in such an epoch, which
allows us to compare its performance to the online algorithms. While
algorithm \RAND\ uses randomization to identify good locations to
serve the current requests, algorithm \DET\ deterministically
migrates to the gravity centers of the demand. Before describing
these two algorithms in detail, we briefly discuss static strategies
without migration.

\subsection{The \STAT\ Algorithm}

In order to compare the benefits of migration to a static scenario,
we derive the competitive ratio of fixed strategies (see
also~\cite{visa10}).
\begin{lemma}\label{lemma:withoutmig}
A system without migration yields a competitive ratio of
$$\rho\in \Theta(\text{Diam}(G)),$$ where $\text{Diam}(G)$ is the
network diameter of substrate network $G$.~\cite{visa10}
\end{lemma}
In a fixed scenario, the best static algorithm \STAT\ hosts $s$ is
in the network center, i.e., the location which minimizes the
worst-case distance traveled by the requests, namely at node $u$ for
which $ u:=\arg\min_{v\in V} \max_{w\in V}\Cost_{\text{acc}}(v,w).$

Note, since there is no migration in the fixed scenario, the
competitive ratio does not depend on any bandwidth constraints
(i.e., on link weights). This means that in networks with highly
heterogeneous links or with links whose capacity changes quickly
over time, a static solution without migration may be good.

\subsection{The \RAND\ Algorithm}

We now describe an online migration algorithm $\RAND$ (see
also~\cite{visa10}). The basic idea of $\RAND$ is to strike a
balance between the request latency cost
$\Cost_{\text{acc}}^{\RAND}$ and the migration cost
$\Cost_{\text{mig}}^{\RAND}$ it incurs, and to continuously move
closer to a possible optimal position. The intuition is that after a
small number of migrations only, either $\RAND$ is at the optimal
position, or an optimal offline algorithm $\OPT$ must have migrated
as well during this time period. Either way, $\OPT$ cannot incur
much smaller costs than $\RAND$. In other words, by using $\RAND$
for moving to good locations in the network, a possible offline
algorithm that migrates less frequently cannot have much lower
access costs than $\RAND$; on the other hand, an offline strategy
with frequent migrations will have similar costs to
$\Cost_{\text{mig}}$.

Let us first consider a scenario with constant bandwidth capacities,
i.e., $\BW(e)=\BW~~\forall e\in E$ and let $\beta=\size(s)/\BW$ be
the corresponding migration cost.
\begin{shaded}
The algorithm $\RAND$ divides time into \emph{epochs}. In each epoch
$\RAND$ monitors, for each node $v$, the cost of serving all
requests from this epoch by a server kept at $v$.  We denote this
counter by $\counter(v)$. $\RAND$ keeps the server at a single node
$w$ till $\counter(w)$ reaches $\beta$.  In this case, $\RAND$
migrates the server to a node $u$ chosen uniformly at random among
nodes with the property $\counter(u) < \beta$. If there is no such
node, $\RAND$ does not migrate the server, and the epoch ends in
that round; the next epoch starts in the next round and the counters
$\counter(v)$ are reset to zero.
\end{shaded}

\begin{lemma}\label{lem:randsingle}
$\RAND$ is $O(\log n)$-competitive in networks with constant
bandwidth.~\cite{visa10}
\end{lemma}
\begin{proof}
Fix any epoch $\Epoch$ and let $\beta$ denote the migration cost. If
$\OPT$ migrates the server within $\Epoch$, it pays $\beta$.
Otherwise it keeps it at a single node paying the value of the
corresponding counter at the end of $\Epoch$.  By the construction
of $\RAND$, this value is at least $\beta$, and thus in either case
$\OPT(\Epoch) \geq \beta$.

The migrations performed by $\RAND$ partition $\Epoch$ into several
\emph{phases}. According to our migration strategy, the access cost
of $\RAND$ in each phase is at most $\beta$. In~\cite{visa10}, we
show that the expected number of migrations within one epoch is at
most $H_n$, where $H_n$ is the $n$-th harmonic number. The number of
phases is then $H_n+1$, and hence $\RAND(\Epoch) \leq \beta \cdot
H_n + \beta \cdot (H_n+1) = \beta \cdot O(\log n)$. This yields the
competitiveness of $\RAND$.
\end{proof}
Note that the analysis does not rely on access costs being measured
as the number of hops. Rather, the analysis (and hence also the
result) is applicable to any metric which ensures that counters
increase monotonically over time, i.e., with additional requests.

For networks with general bandwidths, $\RAND$ can be adopted in such
a way that it migrates when the counter of the current location $v$
reaches $\size(s)/\min_{e}\BW(e)$, that is, when $\counter(v)\geq
\size(s)/\min_{e}\BW(e)$. Thus, the cost of the optimal algorithm in
each epoch is at least $\size(s)/\max_{e}\BW(e)$, while the cost of
$\RAND$ is at most $\size(s)/\min_{e}\BW(e)$. Therefore, by the same
arguments as in the proof of Lemma~\ref{lem:randsingle}, we
immediately obtain the following result.
\begin{theorem}
$\RAND$ is $O(\mu\cdot \log n)$-competitive in general networks,
where $\mu=\max_{e,e'\in E} \BW(e)/\BW(e')$.
\end{theorem}

\subsection{The \DET\ Algorithm}

While for the analysis of \RAND, we assumed an oblivious adversary
which cannot be adaptive with respect to the random choices made by
the online algorithm, we now focus on deterministic algorithms \DET.
As we will see, a logarithmic competitive ratio can also be
achieved. Again, we will first assume $\BW(e)=\BW~~\forall e\in E$
and $\beta=\size(s)/\BW$.

\begin{shaded}
\DET\ divides time into \emph{epochs} consisting of one or multiple
\emph{phases} between which \DET\ migrates. Again, we have counters
$\counter(v)$ for each node $v$ that are set to zero at the
beginning of an epoch. These counters accumulate the access costs of
an epoch if the server was permanently located at $v$. Henceforth,
we will call all nodes $v$ for which at time $t$,
$\counter(v)<\beta/40$, \emph{active nodes} at time $t$. Assume that
algorithm \DET\ is currently at some node $v$. \DET\ remains at this
node until it accumulated there access costs of $\beta$. Then, a new
phase starts, and \DET\ computes the \emph{gravity center} $w$,
i.e., the ``center'' of the currently active nodes. Formally, let
$d$ denote the shortest path metric (w.r.t.~access costs) on the
network $G$. The gravity center of a subset $V' \subseteq V$ of
nodes is defined as the (not necessarily unique) node $\G(V') = \arg
\min_{v \in V'} \sum_{u \in A} d(u,v)$, where $A$ is the set of
access points. (Ties are broken arbitrarily.) \DET\ migrates to $w$
and a new phase starts. If there is no active node left, the epoch
ends.
\end{shaded}

In order to study the competitive ratio of \DET, we exploit the fact
that a request always increases the counter of several nodes besides
the gravity center (namely: a constant fraction) by at least a
certain value (again, a constant fraction) as well.
\begin{lemma}\label{lem:costatothers}
Let $\lambda_1 = 1/5$ and $\lambda_2 = 1/4$. Fix any active set
$V'$. Let $r$ be an arbitrary requesting node (at some step). Assume
the counter at the gravity center $\G(V')$ increased by $F$ because
of this request. Then there are at least $\lambda_2 \cdot |V'|$
nodes from $V'$ whose counters increased at least by~$\lambda_1
\cdot F$.
\end{lemma}
\begin{proof}
Assume the contrary. It means that there are at least $(1-\lambda_2)
\cdot |V'|$ nodes from $V'$ whose counter increase is smaller than
$\lambda_1 \cdot F$. Denote this set by $V''$. We know that the
distance between the request and the center is $d(\G(V'),r) = F$,
and $\forall u\in V''$, $d(u,r)<\lambda_1 \cdot F$. Therefore,
$\forall u,v \in V''$, $d(u,v) < 2 \lambda_1 \cdot F$: the diameter
of the set $V''$ is relatively small.

Now let $\xi$ be any node of $V''$. We show that $\xi$ would be a
better candidate for the gravity center than $\G(V')$ is. Using
triangle inequalities, we obtain
\begin{footnotesize}
\begin{align*}
    \sum_{u \in V'} d(\G(V'),u) \;
        = &\; \sum_{u \in V''} d(\G(V'),u) + \sum_{u \in V' \setminus V''} d(\G(V'),u) \\
        \geq &\; \sum_{u \in V''} \left[ d(\G(V'),r) - d(u,r) \right]\\
         & + \sum_{u \in {V' \setminus V''}} d(\G(V'),u) \\
        > &\; (1-\lambda_1) \cdot |V''| \cdot F + \sum_{u \in {V' \setminus V''}}
        d(\G(V'),u)\\
        >&\;\frac{4}{5}\cdot |V''| \cdot F + \sum_{u \in V' \setminus V''} d(\G(V'),u)
            \enspace
\end{align*}
\end{footnotesize}
\noindent because $d(\G(V'),r)= F$ and $d(u,r)\leq \lambda_1 \cdot
F$, and by substituting $\lambda_1=1/5$. On the other hand, note
that $|V' \setminus V''|\leq |V'|/4\leq |V''|/3$ and
\begin{footnotesize}
\begin{align*}
    \sum_{u \in V'} d(\xi,u) \;
        = &\; \sum_{u \in V''} d(\xi,u) + \sum_{u \in V' \setminus V''} d(\xi,u) \\
        < &\; 2 \lambda_1 \cdot |V''| \cdot F\\ &+ \sum_{u \in V' \setminus V''} \left[ d(\xi,r) + d(r,\G(V')) + d(\G(V'),u) \right] \\
        < &\; 2 \lambda_1 \cdot |V''| \cdot F + |V' \setminus V''| \cdot (1+\lambda_1) \cdot
        F\\
        &+ \sum_{u \in V' \setminus V''} d(\G(V'),u)\\
        \leq &\; \frac{4}{5}\cdot |V''|\cdot F+ \sum_{u \in V' \setminus V''} d(\G(V'),u)
    \enspace
\end{align*}
\end{footnotesize}
\noindent because $d(\xi,r)<\lambda_1 \cdot F$, $d(r,\G(V'))= F$,
and by substituting the value of $\lambda_1=1/5$. This contradicts
that $\G(V')$ is the gravity center of $V'$.
\end{proof}

From Lemma~\ref{lem:costatothers} it follows that when the counter
at the gravity center exceeds a given threshold, the counter of many
nodes besides the center must be high as well.
\begin{lemma}\label{lemma:activeset}
Fix any threshold $\tau$. When the counter at the gravity center
$\G(V')$ exceeds $\tau$, then there exists $V'' \subseteq V'$,
$|V''| \geq \frac{1}{8} \cdot |V'|$, such that for all $v \in V''$,
the counter at $v$ is at least $\tau/40$.
\end{lemma}
\begin{proof}
Assume the contrary. This means that there exists $V'' \subseteq
V'$, $|V''| \geq \frac{7}{8} \cdot |V'|$, such that for all $v \in
V''$, the counter at $v$ is smaller than $\tau/40$. Hence $\sum_{v
\in V''} \counter(v) < |V''| \cdot \tau / 40 \leq |V'| \cdot \tau /
40$. On the other hand, by Lemma~\ref{lem:costatothers}, each time
the counter $\counter(\G(V'))$ increases by $F$, at least $1/4 \cdot
|V'|$ counters from set $V'$ (and hence at least $1/8 \cdot |V'|$
counters from set $V''$, since $|V'\setminus V''| \leq
\frac{|V'|}{8}$) increase by $F/5$. Hence, in this case, the sum of
counters from $V''$ increases at least by $1/40 \cdot |V'| \cdot F$.
Therefore, when $\counter(\G(V')) \geq \tau$, $\sum_{v \in V''}
\counter(v) \geq |V'| \cdot \tau / 40$, which is a contradiction.
\end{proof}

For the competitive ratio, we therefore have the following result.
\begin{theorem}\label{thm:detsingle}
\DET\ is $O(\log{n})$-competitive.
\end{theorem}
\begin{proof}
First, we consider the cost of the optimal offline algorithm. If
\OPT\ migrates in an epoch, it has costs $\beta$. Otherwise, due to
the definition of \DET, as there are no active nodes left at the end
of an epoch, the access costs of any node is also in the order of
$\Omega(\beta)$. Regarding \DET, we know that in each phase, access
costs are at most $\beta$, and it remains to study the number of
phases per epoch. By Lemma~\ref{lemma:activeset}, we know that in
each phase, the number of active nodes is reduced by a factor at
least $1/8$. Therefore, there are at most $O(\log{n})$ many phases
per epoch, and the claim follows.
\end{proof}

Note that we did not try to optimize the constants in this proof,
and in practice (and in our simulations), alternative thresholds can
be applied yielding better (but qualitatively equivalent) results.

Again, for networks with general bandwidths, $\DET$ can be adopted
in such a way that it migrates when the counter of the current
location $v$ reaches $\size(s)/\min_{e}\BW(e)$, that is, when
$\counter(v)\geq \size(s)/\min_{e}\BW(e)$. By the same arguments as
above, this adds a factor $\max_{e,e'\in E} \BW(e)/\BW(e')$ to the
competitive ratio.

\subsection{Remarks}

Recall that the adversarial model for \RAND\ and \DET\ is different,
and hence, one has to be careful when comparing the competitive
ratios: the bound for $\RAND$ only holds against oblivious
adversaries, and we expect the center of gravity approach to perform
better in worst-case scenarios with adaptive adversaries. Our
simulations show that the question which of the two strategies is
more efficient depends on the scenario.

Also note that both \RAND\ and \DET\ are quite general with respect
to the measure of access costs, i.e., the derived bounds hold for
arbitrary latency functions on the links in case of \RAND. This
allows us to generalize our analysis to scenarios where the access
latency, in addition to the sum of the link latencies, depends also
on the capacity of the hosting node: We simply need to take the
capacities into account when increasing the counters. In case of
\DET, the access costs must fulfill the triangle inequality.

\section{Inter-Provider Migration}\label{sec:inter}

The flexibility offered by network virtualization is not limited to
a single PIP. Rather, a Virtual Network Provider may have contracts
with multiple infrastructure providers, and provision a service
across PIP boundaries. In the following, we extend our model to
multiple provider scenarios. In particular, we assume that migrating
a server across a PIP boundary entails a fixed ``roaming'' cost
$\pi$ for each transit. Since we assume that a PIP typically does
not reveal its internal resource structure, we seek to come up with
migration algorithms that pose minimal requirements on the knowledge
of a PIP topology.

In order to study the benefits of migration, we again consider a
scenario without migration (algorithm \STAT). Of course,
Lemma~\ref{lemma:withoutmig} still applies: in a fixed scenario, the
best location for hosting $s$ is in the network center, i.e., the
location which minimizes the distance traveled by the requests.

In the following, we present how the randomized algorithm \RAND\ and
the deterministic algorithm \DET\ can be extended to multi-PIP
scenarios. We consider $k$ PIPs, migration inside a PIP costs
$\beta$, access costs are the number of hops, and migrating across
providers costs $\pi$ per crossed PIP boundary. We will concentrate
on the more realistic case where $\pi \geq \beta$. (If $\pi <
\beta$, our single PIP algorithms could be applied without taking
into account transit costs. This yields a performance in the order
as derived for the case $\pi \geq \beta$.)

It is sometimes useful to think of the \emph{PIP graph}, the graph
where all the nodes of one PIP form one vertex and two PIPs are
connected if there is a connection between nodes of the respective
PIPs in the substrate graph. In particular, we will refer to the
\emph{diameter of the PIP graph}, the largest number of PIPs to be
traversed on a shortest migration path, by $\Delta$.

Algorithm $\PRAND_k$ generalizes $\RAND$ by moving the server to one
of the PIPs having lower costs.
\begin{shaded}
The algorithm $\PRAND_k$ divides time into three types of
\emph{epochs}: \emph{huge epochs} which consist of one or several
\emph{large epochs} which in turn consist of $\lceil \pi/\beta
\rceil$ \emph{small epochs}. For each node $u$, we use two counters
$\counter(u)$ and \counterLarge($u$) to count the access cost during
a small and a large epoch, respectively. At the beginning of a small
epoch, all nodes are \emph{active}; similarly, at the beginning of a
huge epoch, we say that all PIPs are \emph{active}. During a small
epoch, the server is migrated within a single PIP only, until there
is no node $u$ left with access costs smaller than $\beta$:
$\PRAND_k$ monitors, for each node $u$, the cost of serving all
requests from this small epoch by a server kept at $u$; $\PRAND_k$
keeps the server at a single node $u$ till $\counter(u)$ reaches
$\beta$. When this happens, $\PRAND_k$ migrates the server to a node
$v$ chosen uniformly at random among nodes of the current PIP with
the property $\counter(v) < \beta$. If there is no such node,
$\PRAND_k$ does not migrate the server, and the small epoch ends in
that round; the next epoch starts in the next round and the counters
$\counter(u)$ are reset to zero.

After $\lceil \pi/\beta \rceil$ small epochs a large epoch ends.
Then $\PRAND_k$ determines the set of PIPs that contain at least one
node $v$ for which \counterLarge($v$)$<\pi$; all other PIPs become
inactive for the remainder of the current huge epoch. If there are
active PIPs left, $\PRAND_k$ chooses an active PIP uniformly at
random and migrates to an arbitrary node of that PIP; otherwise the
server stays where it is, and a new huge epoch begins.
\end{shaded}

We can derive the following competitive ratio on $\PRAND_k$'s
performance.
\begin{theorem}\label{thm:randmultiPIP}
$\PRAND_k$ is $O(\log{k}\cdot(\log{n_1}+\Delta))$-competitive in
networks with constant bandwidth and $k$ PIPs, where $n_1$ is the
size of the largest PIP, and $\Delta$ is the ``diameter of the PIP
graph''.
\end{theorem}
\begin{proof}
From Lemma~\ref{lem:randsingle}, we know that during a small epoch,
$\PRAND_k$ accumulates a cost of at most $O(\beta \log{n_1})$: There
is at most a logarithmic number of migrations, and the access costs
per phase is at most $\beta$. Recall that a large epoch consists of
at most $\lceil \pi/\beta \rceil$ many small epochs, and
subsequently, a remaining active PIP is chosen uniformly at random.
Thus, similarly to Lemma~\ref{lem:randsingle}, it holds that there
are at most $O(\log k)$ many large epochs, yielding a total access
cost of $O(\log{k} \cdot \log{n_1} \cdot \pi/\beta \cdot
\beta)=O(\pi \log k  \cdot \log n)$. The migration costs within PIPs
are of the same order. The transit costs to move the server between
PIPs amounts to at most $O(\Delta \pi \log{k})$. Thus, the overall
cost of $\PRAND_k$ per huge epoch is in the order of
$O(\pi\log{k}\cdot(\log{n_1}+\Delta))$. On the other hand, an
optimal offline algorithm must have had costs of at least $\pi$ as
well during this huge epoch: if the optimal algorithm migrates
between PIPs, the claim follows trivially. Otherwise, the optimal
offline algorithm is located at a single PIP during the entire huge
epoch; by the construction of $\PRAND_k$, there must exist a large
epoch in which the optimal offline algorithm incurred a cost of at
least $\Omega(\pi)$: per small epoch the (access or migration) costs
are at least $\beta$, and there are $\lceil \pi/\beta \rceil$ many
small epochs in a large epoch. The claim follows.
\end{proof}
Note that the proof of Theorem~\ref{thm:randmultiPIP} is overly
pessimistic, as it assumes several large migration distances that
the optimal offline algorithm can avoid. We believe that $\PRAND_k$
performs better, also in the worst-case, a conjecture that is also
manifested by our experiments.

A similar extension also works for the deterministic variant.
\begin{shaded}
The algorithm $\PDET_k$ divides time into three types of
\emph{epochs}: a \emph{huge epoch} consists of multiple \emph{large
epochs}, and a large epoch consists of $40\lceil \pi / \beta\rceil$
\emph{small epochs}. Again, we use counters $\counter(u)$ to
accumulate the access costs of a node $u$ during a small epoch; in
addition, a counter \counterLarge($u$) is used to accumulate access
costs during a large epoch. In the beginning, all PIPs are set to
\emph{active}. At the beginning of a small epoch, the $\counter(u)$
values are set to zero for all nodes within the current PIP.
$\PDET_k$ then monitors, for each node $u$, the cost of serving all
requests from this small epoch by a server kept at $u$. $\PDET_k$
leaves the server at a single node $u$ till $\counter(u)$ reaches
$\beta$. In this case, $\PDET_k$ migrates the server to a node $v$
which constitutes the center of gravity among the \emph{active}
nodes of the current PIP, i.e., the nodes $w$ of the current PIP for
which it still holds that $\counter(w)<\beta/40$. If there is no
active node left within the current PIP, a  small epoch ends in that
round; the next small epoch starts in the next round. After
$40\lceil \pi / \beta\rceil$ small epochs, a new large epoch starts,
and all nodes $u$ in the network with \counterLarge($u$)$\geq
\pi/40$ become inactive with respect to the large epoch. Among all
remaining active nodes of the large epoch, $\PDET_k$ determines the
center of gravity of all nodes and moves the server to the
corresponding PIP, and a new large epoch begins. Otherwise, if there
is no PIP left that contains active nodes, the server stays where it
is, and a new huge epoch starts.
\end{shaded}

We can show the following result.
\begin{theorem}\label{thm:constant2}
$\PDET_k$ is $O(\log{n}(\log{n_1}+\Delta))$-competitive in networks
with constant bandwidth and $k$ PIPs, where $n_1$ is the size of the
largest PIP, and $\Delta$ is the diameter of the PIP graph.
\end{theorem}
\begin{proof}
First we compute the total cost of $\PDET_k$ in a huge epoch. It
follows from Theorem~\ref{thm:detsingle} that a large epoch consists
of $40\lceil \pi/\beta \rceil$ small epochs of $O(\beta \log{n_1})$
access costs and at a logarithmic number of migrations amounting to
cost $O(\beta \log{n_1})$ as well, yielding a total cost per large
epoch of $O(\pi \log{n_1})$. Now observe that there is at most a
logarithmic number of large epochs per huge epoch: $\PDET_k$
guarantees that the server is not migrated to another PIP as long as
there is a node left in the current PIP with access costs smaller
than $\pi$; in particular, for the center of gravity $u$ of the
current large epoch PIP, \counterLarge($u$)$\geq \pi$, and hence,
again by Theorem~\ref{thm:detsingle}, a constant fraction of nodes
in the entire network must become inactive per large epoch. Summing
up over the large epochs and adding the transit cost of at most
$O(\Delta \pi\log{n})$, the total cost is at most
$O(\log{n}\cdot\pi(\log{n_1}+\Delta))$. The cost of the optimal
offline algorithm can be analyzed similar to the proof of
Theorem~\ref{thm:randmultiPIP}: If the offline algorithm migrates
during a huge epoch, it has a cost of at least $\pi$; otherwise, it
has either access or migration costs of at least $\beta/40$ per
small epoch and hence $\Omega(\pi)$ per large epoch, and the claim
follows.
\end{proof}
Again, we believe that the actual ratio is better, even in the
worst-case, as our analysis is pessimistic.

\subsection{Remarks}

Note that $\PRAND_k$ has the attractive property that it poses
minimal assumptions on the knowledge of the infrastructure topology
and allow for a large autonomy on the PIP level. Typically,
substrate providers are known for their secrecy on traffic matrices
as well as topology information. All information needed by
$\PRAND_k$ is the set of providers that could have served the
requests of a certain time period at lower cost, e.g., the set of
providers that could make ``a better offer''. $\PDET_k$ on the other
hand requires more knowledge of the topology. It assumes that
gravity centers can be computed across PIP boundaries, which is
unrealistic. However, while this facilitates the formal analysis, we
believe that pragmatic implementations that move the service, e.g.,
to the PIP which lies ``at the center'' of the active providers,
yield good approximations and justify the validity of concept and
analysis.

\section{Optimal Offline Algorithms}\label{sec:optoff}

In the competitive analysis of our online algorithms we often argued
about a hypothetical optimal offline algorithm to which we compare
our costs; there was no need to find or describe the offline
algorithm explicitly. However, while the decisions when and where to
migrate servers typically needs to be done \emph{online}, i.e.,
without the knowledge of future requests, there can be situations
where it is interesting to study which migration pattern would have
been optimal \emph{at hindsight}. For example, if it is known that
the requests follow a regular pattern (e.g., a periodic pattern per
day or week), it can make sense to compute an optimal migration
strategy offline and apply it in the future. Another reason for
designing optimal offline algorithms explicitly is that an optimal
solution is required to compute the competitive ratio in our
simulations.

This section is based on the ideas described in the VISA workshop
version of this paper~\cite{visa10}. We present an optimal offline
algorithm for our server migration problems. It turns out that
offline strategies can be computed for many different scenarios, and
we describe a very general algorithm here. Similarly to the online
algorithms, offline strategies can be computed efficiently both for
intra and inter provider migration.

It exploits the fact that migration exhibits an optimal substructure
property: Given that at time $t$, the server is located at a given
node $u$, then the most cost-efficient migration path that leads to
this configuration consists solely of optimal sub-paths. That is, if
a cost minimizing path to node $u$ at time $t$ leads over a node $v$
at time $t'<t$, then there cannot be a cheaper migration sub-path
that leads to $v$ at time $t'$ than the corresponding sub-path.

$\OPT$ essentially fills out a matrix
$opt[\text{time}][\text{node}]$ where $opt[t][v]$ contains the cost
of the minimal migration path that leads to a configuration where
the server satisfies the requests of time $t$ from node $v$.
Assume that initially, the service is located at node $v_0$. Thus,
initially, $opt[0][u]=Cost_{\text{mig}}(v_0,u) + \left[\sum_{v\in
\sigma_0}Cost_{\text{acc}}(v,u) \right]$ as the migration origin is
$v_0$, and as a request needs to travel on the access link from the
terminal to $v$ and from there to $u$ (w.l.o.g., we assume that the
cost $Cost_{\text{acc}}$ contains the first wireless hop from
terminal to substrate network).

For $t>0$, we find the optimal values $opt[t][u]$ by considering the
optimal migration paths to any node $v$ at time $t-1$, and adding
the migration cost from $v$ to $u$. That is, in order to find the
optimal cost to arrive at a configuration with server at node $u$ at
time $t$:
\begin{shaded}
$$
\min_{v\in V}\left[opt[t-1][v]+Cost_{\text{mig}}(v,u)+\sum_{w\in
\sigma_t}Cost_{\text{acc}}(w,u)\right]
$$
\end{shaded}
\noindent where we assume that $Cost_{\text{acc}}$ includes the
first (wireless) hop of the request from the terminal to the
substrate network, and where $Cost_{\text{mig}}(v,v)=0~~\forall v$.

We have the following runtime result.
\begin{theorem}
The optimal offline migration policy $\OPT$ can be computed in
$O(n^3 + n^2 \sum_{t\in \Gamma}|\sigma_t|)$ time, where $\Gamma$ is
the set of rounds in which events occur.
\end{theorem}
\begin{proof}
Note that we can constrain ourselves to optimal offline algorithms
where migration will only take place in ``active'' rounds $\Gamma$
with at least one request. This is useful in case of sparse
sequences with few requests. The $opt[\cdot][\cdot]$-matrix contains
$|\Gamma| \cdot n$ entries. In order to compute a matrix entry, we
need to consider each node $v\in A$ from which a migration can
originate; for each such node, the access cost from all the requests
in $\sigma_t$ need to be computed. Both the shortest access paths
and the migration costs can be looked up in a pre-computed table
(pre-computation in time at most $O(n^3)$, e.g., by
\emph{Floyd-Warshall's algorithm}) and require a constant number of
operations only, which implies the claim.
\end{proof}

Note that $\OPT$ is not an online algorithm. Although $opt[t]$ does
not depend on future requests, in order to reconstruct the optimal
migration strategy at hindsight, the configuration of minimal cost
after the last request is determined, and from there, the optimal
path is given by recursively finding the optimal configuration at
time $t-1$ which led to the optimal configuration at time $t$.

\section{Simulations}\label{sec:simulations}

In order to complement our formal insights and in order to study the
behavior of our algorithms in different environments, we implemented
a simulation framework. In the following, we report on some of our
results in more detail.

\subsection{Set-Up}

We conducted experiments on both artificial \emph{Erd\"os-R\'{e}nyi
graphs} random graphs (with connection probability 1\%) as well as
more realistic graphs taken from the \emph{Rocketfuel
project}~\cite{rocketfuel2,rocketfuel} (including the corresponding
latencies for the access cost).

If not stated otherwise, we assume that link bandwidths are chosen
at random (either T1 (1.544 Mbit/s) or T2 (6.312 Mbit/s)), that the
server size is 2048MB, that $\beta$ equals the server size divided
by the average bandwidth, and $\pi=3\beta$.

Note that our the runtime of the optimal offline algorithm and hence
the computation of the competitive ratio is expensive in large
networks; therefore, the scale of our experiments is typically
limited. However, as our online algorithms have a much lower runtime
than $\OPT$, experiments that do not rely on optimal offline results
can be conducted for much more nodes. Moreover, to gain insights
into the behavior of our algorithms in networks of this size, we use
a threshold $\tau=1/3$ (rather than $\tau=1/40$) to inactivate nodes
in $\DET$. This value is more practical and does not change the
qualitative results for large networks.

As our real traffic patterns are subject to confidentiality, we
consider two different simplified, artificial scenarios. Our
scenarios assume that the substrate topology does not reflect the
geographic situation or user pattern at all. This is conservative of
course, and online migration algorithms typically perform better if
requests move along the topology.

\begin{shaded}
\textbf{Time Zones Scenario:} This scenario models an access pattern
that can result from global daytime effects. We divide a day into
$T$ time periods. At each time $t$, $p\%$ of all requests originate
from a node chosen uniformly at random from the substrate network
(pessimistic assumption). The sojourn time of the requests at a
given location is distributed exponentially with parameter $\lambda$
as well.  In addition, there is background traffic: the remaining
requests originate from nodes chosen uniformly at random from all
access points.
\end{shaded}

We also studied an alternative scenario, capturing traffic from
commuters.
\begin{shaded}
\textbf{Commuter Scenario:} This scenario models an access pattern
that can result from commuters traveling downtown for work in the
morning and returning back to the suburbs in the evening. We use a
parameter $T$ to model the \emph{frequency} of the changes. At time
$t$ mod $T < T/2$, there are $2^{t \text{ mod } T}$  requests
originating from access points chosen uniformly at random around the
center of the network. In the second half of the day, i.e., for $t
\in [T/2,...,T]$, the pattern is reversed. Then a new day starts.
The commuter scenario can come in different flavors, e.g., where the
total number of requests remains constant over time, or where the
load is changing. We use the static load scenario in our
simulations. The total number of requests per round is fixed to
$2^{T/2}$. At time $t_i < T/2$, the requests originate from $p =
2^{t_i \text{ mod } T}$ of all access points including the network
center ($2^{T/2}/p$ requests per access point), until single
requests originate from $2^{T/2}$ access points. Then, the same
process is reversed until all $2^{T/2}$ requests originate from a
single access point: the network center. We assume that the time
period between $t_i$ and $t_{i+1}$ is distributed exponentially with
parameter $\lambda$.
\end{shaded}

\subsection{Intra Provider Migration}

A first set of experiments studies the competitive ratio as a
function of the number of nodes.
Figure~\ref{fig:single_pip_det_time_zone} reports on the impact of
different correlations of the requests in the time zone scenario.
First, we can observe that the competitive ratios are generally
quite low, and more or less independent of the network size. In
order to take into account that larger networks typically come with
a larger requests set, we assume that the number of requests per
round is one fifth of the network size. We can see that the
competitive ratios of $\DET$ are again quite low, but the optimal
offline algorithm can do relatively better if $p$ is large, which
meets our intuitions.
\begin{figure} [t]
\begin{center}
\includegraphics[width=0.85\columnwidth]{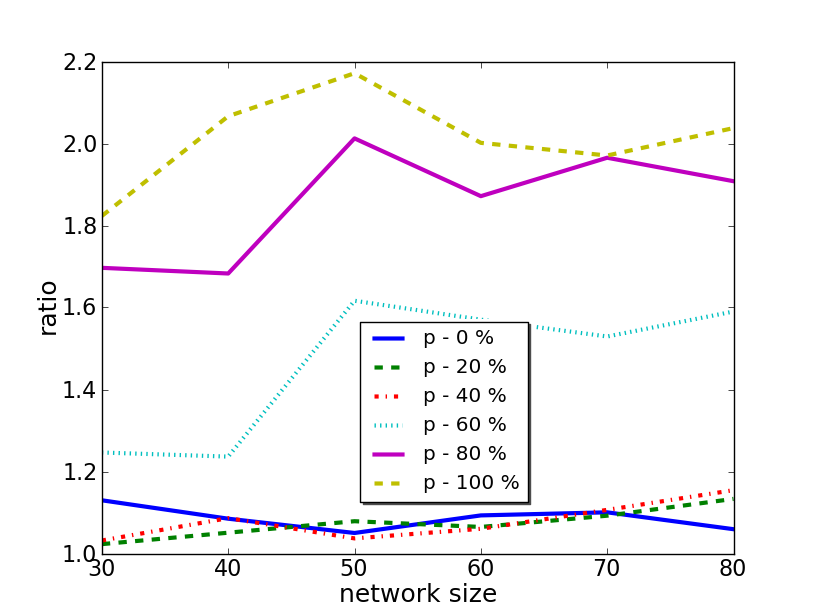}\\
\caption{Competitive ratio of $\DET$ as a function of network size
and $p$ in a time zone scenario. Results are averaged over 10 runs,
we use $\lambda=10$ and collected data over a period of 400
rounds.}\label{fig:single_pip_det_time_zone}
\end{center}
\end{figure}

The same results for $\RAND$ can be seen in
Figure~\ref{fig:single_pip_rand_time_zone}. Again, larger $p$ yield
higher ratios, and generally the performance is worse than the one
of $\DET$.
\begin{figure} [t]
\begin{center}
\includegraphics[width=0.85\columnwidth]{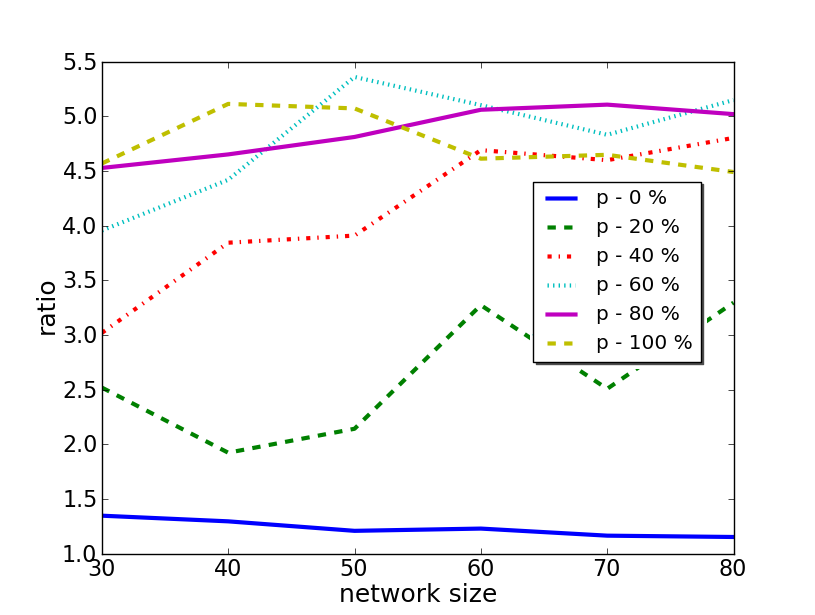}\\
\caption{Same experiment as
Figure~\ref{fig:single_pip_det_time_zone} for
$\RAND$.}\label{fig:single_pip_rand_time_zone}
\end{center}
\end{figure}


Due to the high time complexity of the optimal offline algorithm, we
conducted some experiments with absolute costs only, see
Figure~\ref{fig:single_pip_online_vs_static}. Again, the number of
requests per round is one fifth of the network size. This reveals an
interesting phenomena, namely that in this time zone scenario,
$\RAND$ becomes better and gets close to the performance of $\DET$
in large networks.
\begin{figure} [t]
\begin{center}
\includegraphics[width=0.85\columnwidth]{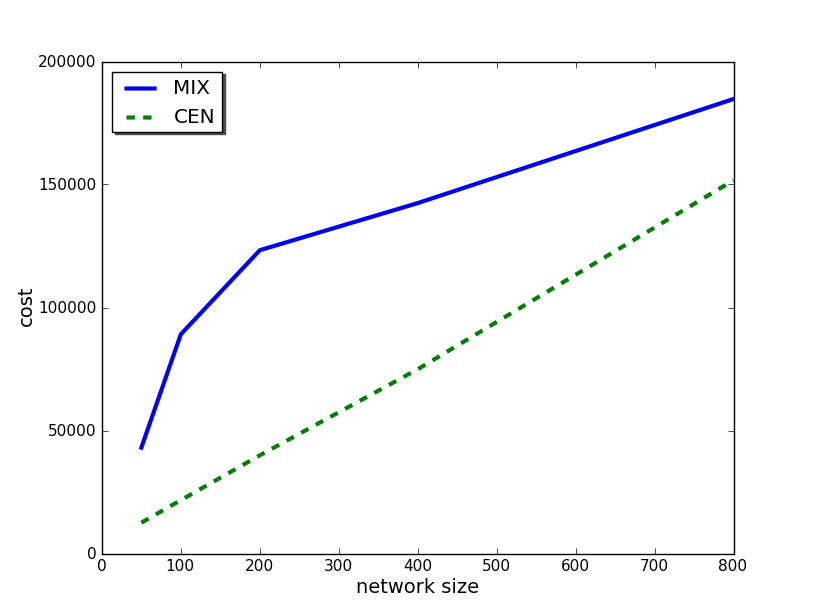}\\
\caption{Absolute costs in time zone scenario: $p=40\%$,
$\lambda=10$, averaged over 10 runs and studying a runtime of 500
rounds.}\label{fig:single_pip_online_vs_static}
\end{center}
\end{figure}

Another interesting question regards how the competitive ratio
depends on the dynamics, i.e., on $\lambda$.
Figure~\ref{fig:ratio_vs_lambda_single_pip} presents our results for
$\RAND$ and $\DET$. Although the variance is quite high (this is
typical, especially for $\RAND$), we can see a trend that in case of
very high dynamics and very low dynamics, the competitive ratio is
slightly lower than for $\lambda$ values in-between ($\lambda$ is
the mean stay duration). This can be explained by the fact that
$\OPT$ can optimize relatively more in scenarios where the online
migration decisions are not obvious.
\begin{figure} [t]
\begin{center}
\includegraphics[width=0.85\columnwidth]{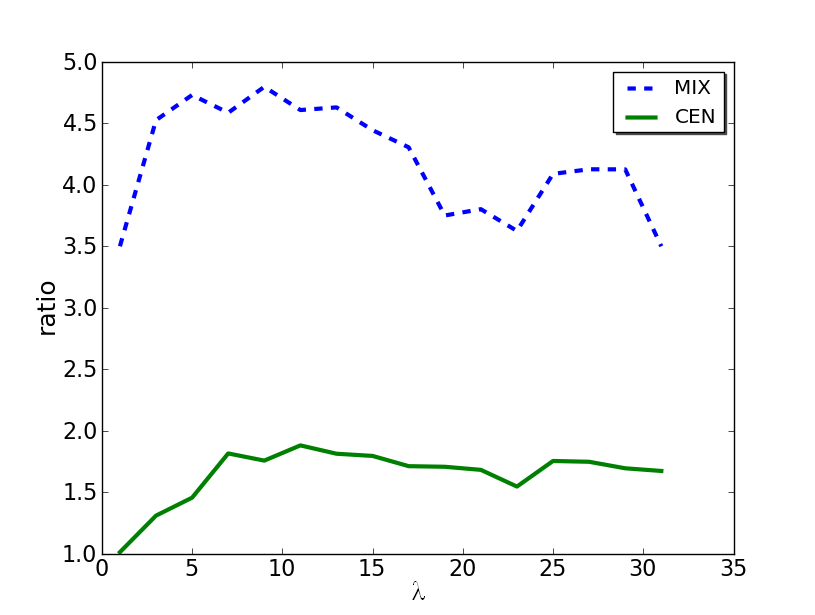}\\
\caption{Competitive ratio in time zone scenario (with $p=60\%$) as
a function of $\lambda$ and averaged over 10 runs and in a network
of 60 nodes. We ran the experiment for 200
rounds.}\label{fig:ratio_vs_lambda_single_pip}
\end{center}
\end{figure}

Generally, in the random graphs of the size used in our experiments,
the diameter is relatively low and hence $\STAT$ typically has a
good performance as well. In the commuter scenario, it is worse than
in the time zone scenario.
Figure~\ref{fig:single_pip_commuter_freq_3} shows that both online
algorithms perform very good (and better than in the time-zone
scenario).
\begin{figure} [t]
\begin{center}
\includegraphics[width=0.85\columnwidth]{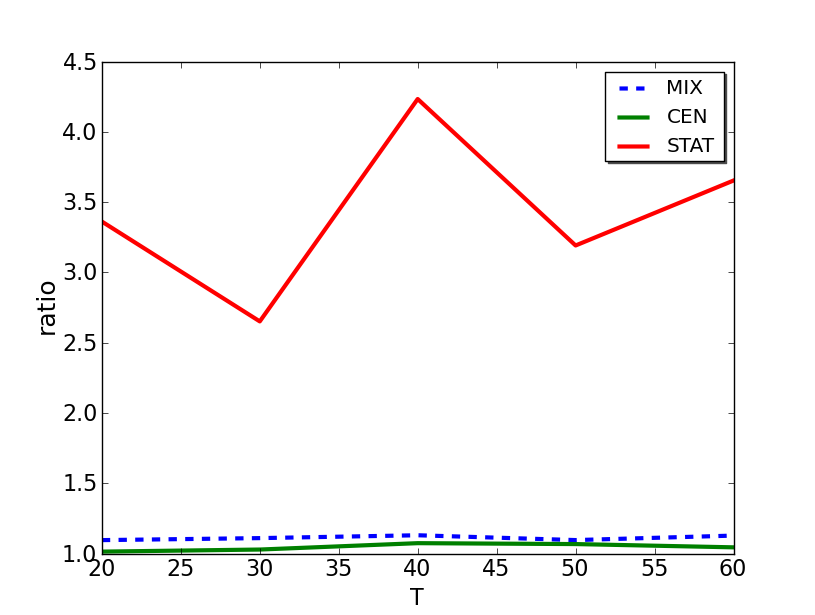}\\
\caption{Commuter scenario with $T=3$, server size 30MB
$\lambda=10$, runtime 200 rounds and averaged over 20
runs.}\label{fig:single_pip_commuter_freq_3}
\end{center}
\end{figure}
One takeaway from our experiments with the different scenarios is
that migration is relatively more beneficial in these instances of
the commuter scenario compared to the time zone instances studied.



Finally, a note on experiments on Rocketfuel topologies. Also there,
the competitive ratios were low. For instance, on the autonomous
system AS-7018 of \texttt{ATT} which consists of 115 nodes (see
Figure~\ref{fig:plot-rocketfuel_singlepip}), and averaging over 50
runs at a runtime of 100, the optimal offline cost is 477.905648298,
\RAND\ has cost of 1179.85 and \DET\ has cost 825.81, which implies
a competitive ratio is around 2.4 and 1.7, respectively. Also \STAT\
performs quite well: without migration the cost lies between \RAND\
and \DET: 1091.51.
\begin{figure} [t]
\begin{center}
\includegraphics[width=0.85\columnwidth]{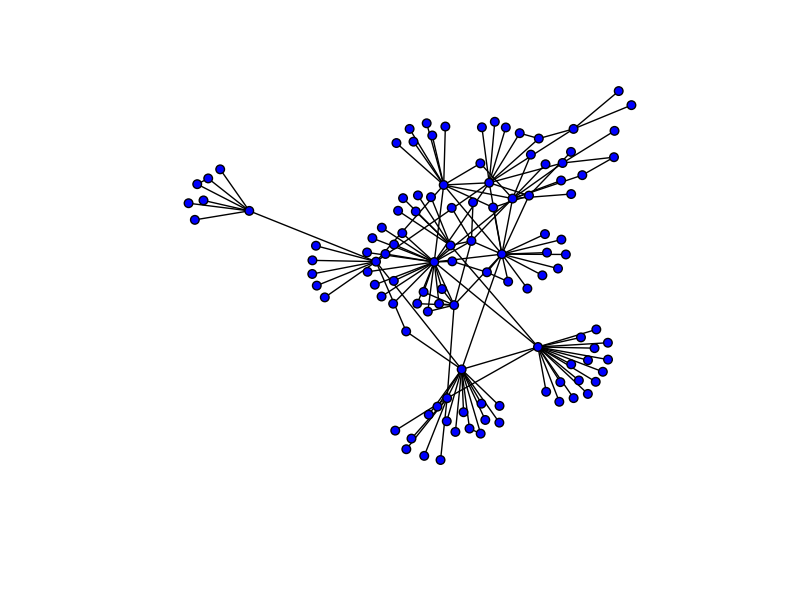}\\
\caption{\texttt{ATT} network AS-7018 with 115
nodes.}\label{fig:plot-rocketfuel_singlepip}
\end{center}
\end{figure}

\subsection{Inter Provider Migration}

We also conducted some experiments with multiple PIPs. Generally, we
used the same random topologies to model a single PIP network, and
connected different PIPs in a circular manner using a random
connection between adjacent providers.

As in the intra provider scenario, we first report on scalability.
Figure~\ref{fig:ratio_vs_no_of_nodes_multipip} plots the competitive
ratio as a function of the total number of nodes per provider, given
that there are three providers. The ratios are similar to the single
PIP case, $\PRAND_k$ is slightly worse than $\PDET_k$, and constant
bandwidth scenarios yield lower ratios.
\begin{figure} [t]
\begin{center}
\includegraphics[width=0.85\columnwidth]{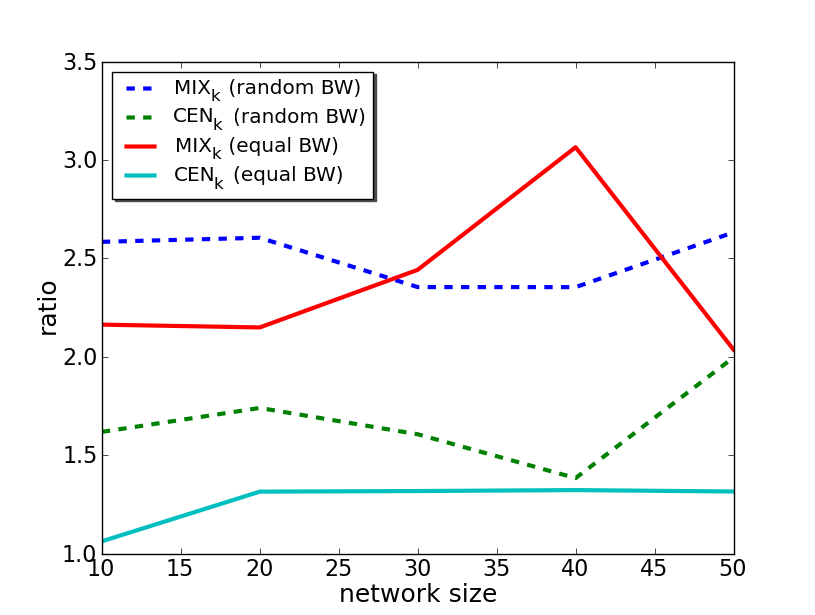}\\
\caption{Competitive ratio in time zone scenario (with $p=0\%$) with
three providers, server size 30MB, $\lambda=2$ and a runtime of 50
rounds. We average the ratio over five
runs.}\label{fig:ratio_vs_no_of_nodes_multipip}
\end{center}
\end{figure}

Figure~\ref{fig:multi_pip_det_time_zone} shows the effect of
different correlation ($p$ values) for $\PDET_k$, and
Figure~\ref{fig:multi_pip_rand_time_zone} studies the analogous
situation for $\PRAND_k$.
\begin{figure} [t]
\begin{center}
\includegraphics[width=0.85\columnwidth]{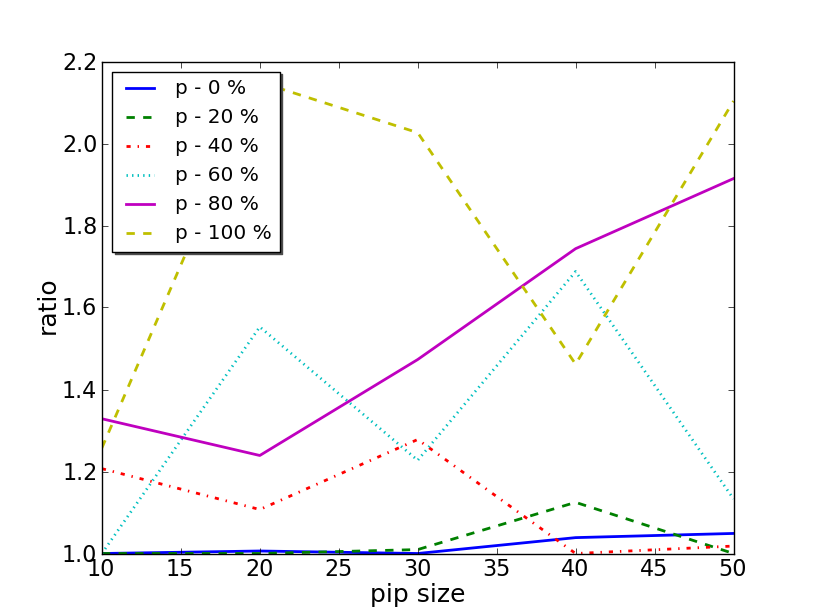}\\
\caption{Time zone scenario with three PIPs, $\lambda=5$, runtime
200 rounds: competitive ratio of $\PDET_k$ as a function of provider
size and $p$.}\label{fig:multi_pip_det_time_zone}
\end{center}
\end{figure}
\begin{figure} [t]
\begin{center}
\includegraphics[width=0.85\columnwidth]{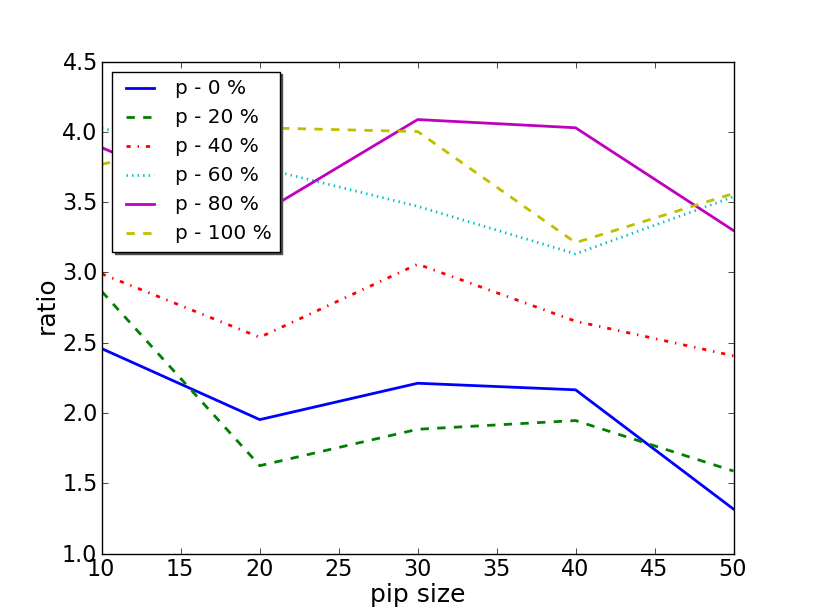}\\
\caption{Like Figure~\ref{fig:multi_pip_det_time_zone} but for
$\PRAND_k$.}\label{fig:multi_pip_rand_time_zone}
\end{center}
\end{figure}



In terms of dynamics, we also have a similar picture as in the
single provider case, see
Figure~\ref{fig:ratio_vs_lambda_multi_pip}.
\begin{figure} [t]
\begin{center}
\includegraphics[width=0.85\columnwidth]{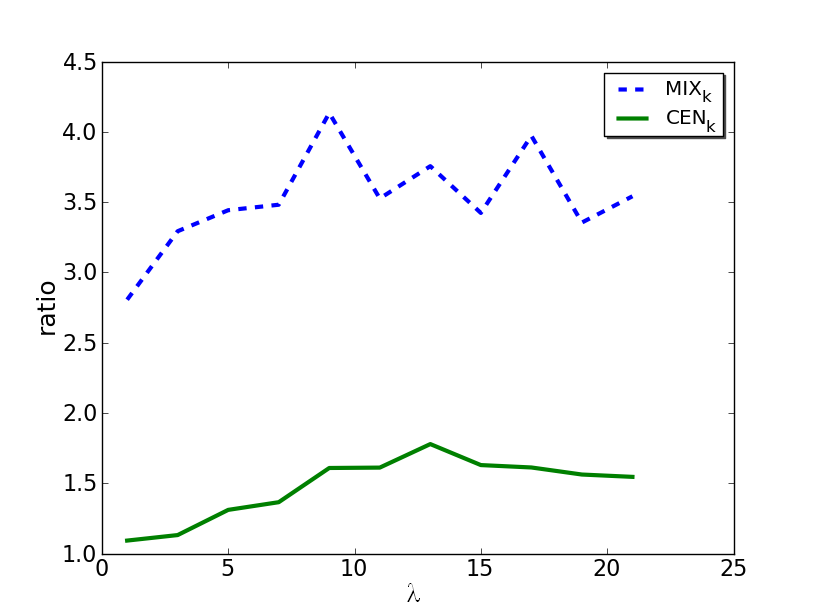}\\
\caption{Time zone scenario ($p=60\%$) for three providers with 20
nodes each, runtime 200 and averaged over 10
runs.}\label{fig:ratio_vs_lambda_multi_pip}
\end{center}
\end{figure}

It turns out that both $\PRAND_k$ and $\PDET_k$ are relatively
robust to $x=\pi/\beta$, the relative cost of transit compared to
migration, although $\PRAND_k$ slightly benefits from lower
migration costs, as we would expect (see
Figure~\ref{fig:multipip_ratio_vs_x}).
\begin{figure} [t]
\begin{center}
\includegraphics[width=0.85\columnwidth]{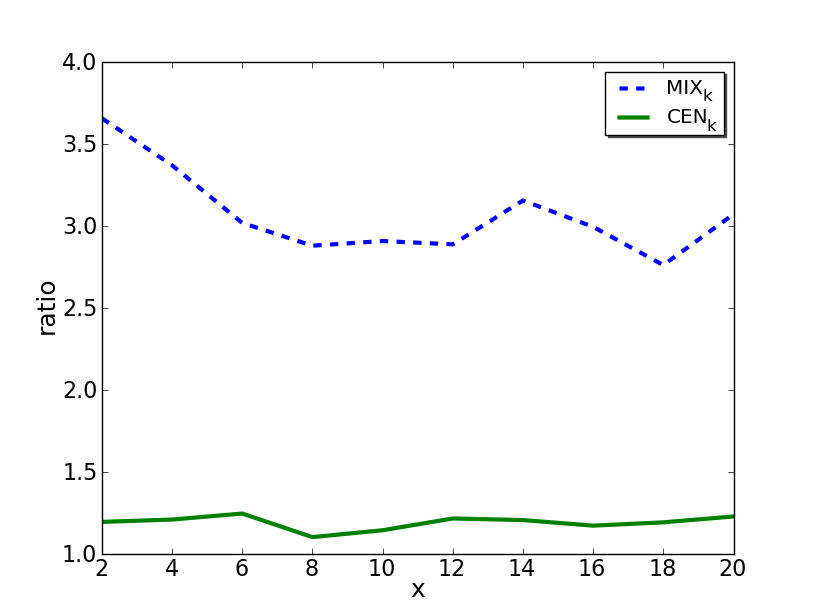}\\
\caption{Ratio as a function of $x=\pi/\beta$. Time zone scenario
($p=50\%$) with 3 PIPs, runtime 400, PIP size 20 nodes, $\lambda=5$,
averaged over 10 iterations.}\label{fig:multipip_ratio_vs_x}
\end{center}
\end{figure}

Finally, we have studied the competitive ratio as a function of the
total number of PIPs. Interestingly, as can be observed in
Figure~\ref{fig:ratio_vs_nop}, a medium number of providers is
slightly worse than scenarios with very few or many  providers, but
the ratio is pretty robust here as well.
\begin{figure} [t]
\begin{center}
\includegraphics[width=0.85\columnwidth]{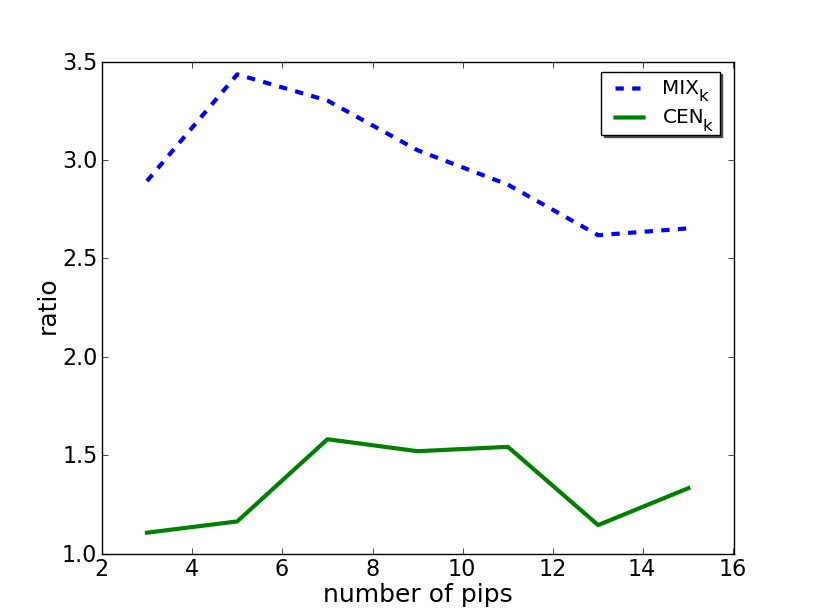}\\
\caption{Time zone scenario ($p=50\%$) runtime 400, PIP size 10
nodes, $\lambda=5$, averaged over 10
iterations.}\label{fig:ratio_vs_nop}
\end{center}
\end{figure}

In an experiment with Rocketfuel topologies (network shown
in~Figure~\ref{fig:plot-rocketfuel_multipip}), averaging over ten
runs, we obtain a total cost of 105.86 for \OPT, 356.89 for \RAND,
and 226.66 for \DET, yielding a competitive ratio of around four for
\RAND\ and two for \DET.
\begin{figure} [t]
\begin{center}
\includegraphics[width=0.85\columnwidth]{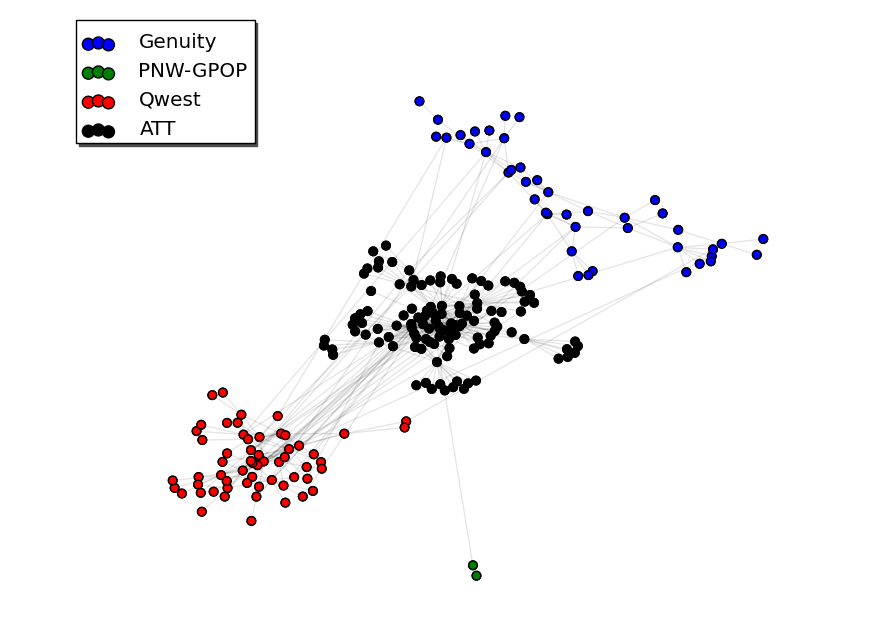}\\
\caption{Multi-PIP network with four U.S.~providers (217 nodes in
total) \texttt{Genuity} (AS-1), \texttt{PNW-GPOP} (AS-101),
\texttt{Qwest} (AS-209), and \texttt{ATT}
(AS-7018).}\label{fig:plot-rocketfuel_multipip}
\end{center}
\end{figure}

\section{Conclusion}\label{sec:conclusion}

At the heart of network virtualization lies the ability to react to
changing environments in a flexible fashion. In order to optimally
exploit the benefits from virtualization, algorithms need to be
designed that adapt dynamically to the current demand; this is
typically difficult as future demand is hard to predict. We believe
that competitive analysis is an important tool to devise and
understand such online algorithms.

This paper studied the cost-benefit tradeoff of online migration in
a system supported by network virtualization, and compared our
system to a setting without migration. We derived the first
migration algorithms, both for inter and intra provider scenarios,
which are competitive even in the worst-case.

We understand our work as a first step towards a better
understanding of competitive virtual service migration, and there
are several interesting directions for future research. For
instance, we believe that the bounds on the competitive ratio for
multiple PIPs are overly pessimistic and can be improved. We also
plan to study (simplified versions of) our algorithms in the wild,
i.e., in our prototype~\cite{visa09virtu} architecture.

Finally, we emphasize that while our formal considerations may give
insights into the benefits of this new technology, e.g., in terms of
improved quality of service, whether and how mobile network provider
will adapt such an approach also depends on many economic factors
that are not taken into account in our model.

\section*{Acknowledgments}

A preliminary version of this paper appeared at the 2010 ACM SIGCOMM
workshop VISA~\cite{visa10}.

Part of this work was performed within the 4WARD project, which is
funded by the European Union in the 7th Framework Programme (FP7),
the Virtu project, funded by NTT DOCOMO Euro-Labs and the
Collaborative Networking project funded by Deutsche Telekom AG. We
would like to thank our colleagues in these projects for many
fruitful discussions; in particular: Dan Jurca (now at Huawei
Technologies Duesseldorf GmbH), Wolfgang Kellerer, Ashiq Khan,
Kazuyuki Kozu, and Joerg Widmer (now at Institute IMDEA Networks).

Special thanks go to Ernesto Abarca who was a great help during the
prototype implementation. We also thank Johannes Grassler and Lukas
W\"ollner for their help with the prototype and the migration
demonstrator. M.~Bienkowski is supported by MNiSW grants number N
N206 368839, 2010--2013 and N N206 257335, 2008--2011.

\bibliographystyle{abbrv}
\bibliography{migration}

\end{document}